\documentclass{article}
\usepackage{ijcai17}
\usepackage{amssymb}
\usepackage{helvet}
\usepackage{courier}

\usepackage{times}

\usepackage{multicol}
\usepackage{multirow}

\usepackage{booktabs}
\usepackage{algorithm}
\usepackage{setspace}
\usepackage{algpseudocode}
\usepackage{amsmath}
\usepackage{amsthm}
\usepackage{bm}
\usepackage{amssymb}
\usepackage{graphicx}
\usepackage{mathtools}
\usepackage{tikz}
\usepackage{times}
\usepackage{subfigure}
\usepackage[small]{caption}
\usetikzlibrary{calc,shapes}

\newtheorem{theorem}{Theorem}

\title{Softpressure: A Schedule-Driven Backpressure Algorithm for Coping with Network Congestion}

\author{Hsu-Chieh Hu \and Stephen F. Smith\\
 Carnegie Mellon University\\ Pittsburgh, PA, USA\\
hsuchieh@andrew.cmu.edu,sfs@cs.cmu.edu
 }

\begin{document}

\maketitle
\begin{abstract}
We consider the problem of minimizing the delay of jobs moving through a directed graph of service nodes.  In this problem, each node may have several links and is constrained to serve one link at a time.  As jobs move through the network, they can pass through a node only after they have been serviced by that node. The objective is to minimize the delay jobs incur sitting in queues waiting to be serviced.  Two distinct approaches to this problem have emerged from respective work in queuing theory and dynamic scheduling: the backpressure algorithm and schedule-driven control.  In this paper, we present a hybrid approach of those two methods that incorporates the stability of queuing theory into a schedule-driven control framework.  We then demonstrate how this hybrid method outperforms the other two in a real-time traffic signal control problem, where the nodes are traffic lights, the links are roads, and the jobs are vehicles.  We show through simulations that, in scenarios with heavy congestion, the hybrid method results in $50\%$ and $15\%$ reductions in delay over schedule-driven control and backpressure respectively. A theoretical analysis also justifies our results.

\end{abstract}

\section{Introduction}
Consider the following constrained network scheduling problem: We are given a networked system comprised of some number of interconnected scheduling agents. Each agent is responsible for managing a set of competing dynamic flows that must share a common resource, e.g., vehicle flows through an intersection in a transportation network or message traffic through servers in a communications network. We assume that agents have multiple queues to accommodate the competing flows. The objective is to minimize expected delay (or maximize expected throughput) of the entities moving through the network. Such a constrained optimization framework in stochastic systems is general and has practical applications in such areas as supply chain management, computer networks and traffic signal control. However, the question of how to most effectively schedule such networked systems still remains unclear.

We assume that the networked system operates in a dynamic environment. Jobs arrive continuously and are unknown before they arrive \cite{pinedo2015scheduling}. To behave intelligently over a networked system, individual agents (nodes) must coordinate scheduling decisions with their neighbors over an extended planning horizon. Recent work in decentralized, online planning has produced techniques for generating traffic signal timing plans with order-of-magnitude longer horizons than was previously possible \cite{Xie2012,xie2012schedule}, and an approach to achieving network-level coordination through exchange of schedule information. Under this {\em schedule-driven approach} new jobs from neighbor agents(and agents further upstream) are anticipated before they arrive, and are continually added into the ongoing problem-solving process in a rolling horizon fashion. The effectiveness of this approach has been shown in an actual urban environment across a range of traffic conditions \cite{smith2013smart}. At the same time, such a dynamic scheduling approach remains susceptible to sub-optimal long-run solution quality due to optimization over only a relatively short time horizon. Thus, decisions made during one static horizon can have unforeseen effects on a later period.

The network scheduling problem is also known as resource allocation problem \cite{tassiulas1992stability,georgiadis2006resource,shah2008network,shah2011message} or stochastic network optimization \cite{neely2010stochastic} in the literature of queueing systems. The most important long-run performance measure considered in queueing theory is stability. Informally, the queueing network is considered to be stable if the queues remain bounded over time. In this sense, queueing theory reasons more about long-run stochastic system characteristics such as stability of system or throughput, whereas dynamic scheduling typically deals with short-term combinatorics and uncertainty. With regards to ensuring stability of a stochastic network, one well recognized approach is the {\em Backpressure} algorithm \cite{tassiulas1992stability}, which has been applied widely in communication networks. Backpressure provides activation policies that chose , among a given set of queues (flows), which to serve at any point based on the queue lengths of all nodes. It guarantees stability at the expense of missing short-term opportunities optimize performance.

Recent work has proposed that hybridization of scheduling and queueing techniques can provide synergistic benefits in different application settings \cite{terekhov2012long,terekhov2014integrating}. Following this view, we propose a hybrid approach to constrained network scheduling problems that integrates backpressure into a dynamic scheduling framework. In this work, we show that the fundamental queueing theory concept of stability can be used to enhance the performance of dynamic scheduling, and that the coupling produces a composite algorithm that outperforms either approach individually. In particular in high congestion situations, where the number of jobs approaches the buffer limits of interconnected queues in the network, the stability of queues becomes more crucial, and the problem becomes less of a scheduling problem. 

In order to stabilize queues in a network (i.e., to prevent unbounded growth of queues), the jobs associated with longer queues or larger incoming arrival traffic should be serviced first. Within the above network scheduling problem, one straightforward way of achieving this behavior is to assign higher weights (i.e., higher priority) to these input jobs and compute schedules that minimize weighted cumulative delay. To balance the emphasis placed on queue management as a function of network saturation, we propose to use queue-length information to establish the weights. In situations where queue lengths are small, jobs causing the larger cumulative delay are serviced first as before; however as the network becomes saturated and queues become longer, jobs associated with longer queues will begin to dominate the original dynamic scheduling. 

To demonstrate and validate our hybrid approach, we focus on a real-world traffic network control problem and extend the schedule-driven approach to real-time  traffic control mentioned earlier \cite{Xie2012,xie2012schedule}. Simulation results demonstrate the ability of our approach to effectively integrate the long-run stability of queueing systems with coordinated, network-level dynamic scheduling. It is shown that, in scenarios with heavy congestion, the hybrid method outperforms both schedule driven control and backpressure individually, resulting in delay reductions of $50\%$ and $15\%$ respectively. Likewise, under low demand levels, it still seen to outperform both individual approaches. Hence, the hybrid approach is seen to effectively combine the relative strengths of schedule-driven control and backpressure. In addition, a theoretical analysis based on Lyapunov-Foster theory also justifies the results \cite{tassiulas1992stability}.

The remainder of the paper is organized as follows. We first introduce the definition of stability and the backpressure algorithm. Next, the algorithm necessary to achieve stabilized queues and its theoretical analysis are discussed. Then, an empirical analysis of the composite approach  is presented. Finally, we discuss the implications of integrating scheduling with queueing stability and conclusions are drawn.
     

\section{Problem Setting}
Consider a dynamic environment where the jobs change dynamically over time and their processing times are affected by various types of uncertainty. 
\subsection{A Queueing Network} \label{qnet}
The system of interest is a queueing network \cite{gross2008fundamentals}.  The connectivity of the networked system is represented by a directed graph $G = (V,E)$, where $V$ is the set of nodes and $E$ is the set of links. We consider a network consisting of $|V| = L$ nodes and $|E| = N$ links. Each node has a scheduling agent to serve jobs belonging to specific classes. On the link $(i,j)$, the node $j$ has a corresponding queue $Q_{ij}$ to buffer approaching jobs. The scheduling agent can only serve one queue at a time. For instance, Figure~\ref{net} shows that node $C$ has two queues $Q_{BC}$ and $Q_{AC}$, and at any point it can either serve $Q_{BC}$ or $Q_{AC}$ (but not both). Jobs arrive over time via an arrival process determined by the last visited node. A job may enter the network from any node and leave the network if it reaches its destination by appropriate routing through the network. We assume that there are $K$ job classes. At each time $t$ a job $i$ belonging to class $k$ arrives at node $l$ with probability $p_k$, and its service rate and processing time on the agent are denoted by $\mu_{lk}$ and $d_{li}$. We assume the agents in the network are able to serve jobs of all classes.  The objective of the agents is to minimize the total cumulative delay of jobs traveling through the network over a given time period, while maintaining the stability of the network.

\begin{figure}[!htbp]
\centering
\includegraphics[scale = 0.4]{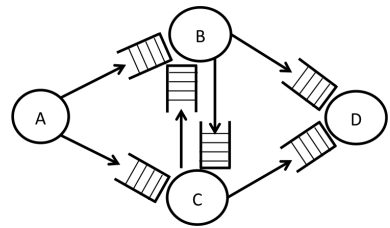}
\caption{An example of queueing networks consisting of  scheduling agents $A, B,C$ and $D$.}
\label{net}
\end{figure}

\subsection{Real-World Example:Transportation Network}
Recent work in decentralized, online planning has developed a schedule-driven approach to real-time traffic control \cite{Xie2012,xie2012schedule}. Key to this approach is a formulation of the core intersection scheduling problem as a single machine scheduling problem, where input jobs are clusters of vehicles in close proximity to each other as they pass through the intersection  (i.e., approaching platoons, queues). This aggregate representation allows long-horizon plans to be efficiently generated and enables network-level coordination through exchange of schedule information. Under this approach, the goal is to allocate green time to different signal  \textit{phases}, over time, where a signal phase is a compatible traffic movement pattern (e.g., East-West traffic flow). Each intersection agent asynchronously computes a schedule of green phases that minimizes the cumulative delay through the intersection of all approaching vehicles, and then communicates expected outflows to its downstream neighbors as it begins to execute its schedule.

The clusters that characterize the waiting and approaching traffic for each signal phase become the jobs that must be sequenced through the intersection. These cluster sequences provide short-term variability of traffic flows at each intersection and preserve the non-uniform nature of real-time flows. Specifically, the input is an ordered sequence of $(\# \textit{ of vehicles}, \textit{arrival time}, \textit{departure time} ) = (|c|, arr, dep)$ triples reflecting each approaching or queued vehicle on each road segment that has been sensed through the intersection's detectors.  Matching the assumptions of the above queueing network setting, all non-competing queues corresponding to a single phase are merged into a single large "virtual" queue.

Once the cluster sequences of each approaching road segments are represented, each cluster is viewed as a non-divisible job and a forward-recursion  dynamic programming search is performed to generate a phase schedule that minimizes the cumulative delay of all clusters in the current prediction horizon. The process constructs an optimal sequence of clusters that maintains the ordering of clusters along each road segment, and each time a phase change is implied by the sequence, then a delay corresponding to the intersection's yellow/all-red changeover time constraints is inserted.  If the resulting schedule is found to violate the maximum green time constraints for any phase (introduced to ensure fairness), then the first offending cluster in the schedule is split, and the problem is re-solved. 

More precisely, the delay that each cluster contributes to the cumulative delay $\sum_c d(c)$ is defined as 
\begin{equation}
d(c) = |c| \cdot (ast - arr(c)),
\end{equation} 
where $ast$ is the actual start time determined by the process through a forward recursion. The optimal sequence (schedule) is the one that incurs minimal delay for all vehicles.

\section{Stability}

In \cite{terekhov2012long}, stability was first introduced to the scheduling community. Informally, a system is considered to be stable if the queues remain bounded over time. This previous work used the definition that the load of each machine, defined as the ratio of the arrival rate to the service rate, is strictly less than 1 as the necessary conditions for stability. However, queueing networks need a more general definition of stability to describe the behaviors of networks. 

\subsection{Network Stability}
To discuss stability for networks, we specify the general network model. We consider that a set of scheduling agents controls job flows through a network with the goal of minimizing the average delay, subject to the constraint that only a specified set of links can be activated simultaneously due to the need to share the common resource. The network is assumed to operate in slotted time, i.e., $t\in\{0,1,2,\dots\}$.  We assume there are $N$ queues in the network. Let $\mathbf{Q}(t) = (Q_1(t), \dots, Q_N(t)) \in \mathbb{R}^N_+$, $t = 0,1,2,\dots$ be the queue length vector of the network, in units of job processing time. In this paper, we adopt the following notion of queueing stability for a network \cite{tassiulas1992stability,georgiadis2006resource,wongpiromsarn2012distributed}: 
\begin{equation}
\bar{Q} \equiv \lim_{T\rightarrow\infty}\frac{1}{T}\sum_{t = 0}^{T-1}\sum_{i = 1}^{N}\mathbb{E}[Q_i(t)] < \infty,
\label{stb}
\end{equation}
implying that the time-averaged length of queues is bounded. This definition also implies that the Markov process that describes the dynamics of the system is \textit{positive recurrent} \cite{meyn2012markov}. In queueing theory, establishing network stability is considered to be a prerequisite to more detailed analysis and scheduling policy design.  

\subsection{Backpressure}
We start with a description of the backpressure algorithm introduced by Tassiulas and Ephremides \cite{tassiulas1992stability}. At each time slot $t$, the agent at each node selects a particular non-conflicting set of links (e.g., a particular signal phase) to serve from the set of all incoming links. Each agent assigns a weight to each such set of links by summing the queue length that the set of links proposes to serve, and then chooses the set of links with the largest weight. It has been shown that the backpressure is stable and throughput optimal for the general networks considered in Section \ref{qnet}. The backpressure algorithm is also known to induce a reasonable (polynomial in network size) average queue-size for this model, and therefore the induced delay is upper bounded. 

The algorithm is sketched as follows: Consider a queueing network with the queue vector $\mathbf{Q}(t)$. Let $Q_{i}(t)$ be the size of queue $i$ at the beginning of time slot $t$. We denote the feasible set of non-conflicting links by $\mathcal{S}\subset \mathbb{R}_+^N$. In every time slot an activation vector $\pmb{\pi} \in \mathcal{S}$ is chosen; $Q_{i}(t)$ is given an amount of service $\pi_{i}$ in that time slot. For simplicity, we will restrict ourselves to $\mathcal{S}$ such that $\mathcal{S} \subset \{0,1\}^N$; that is, for any $\pmb{\pi} \in \mathcal{S}$ , $\pi_{i} = 1$ ($Q_{i}(t)$ receives one unit of service) or 0 ($Q_{i}(t)$ receives no service). The backpressure algorithm chooses a vector $\pmb{\pi}$ such that 

\begin{equation}
\pmb{\pi} \cdot \mathbf{Q}(t)  = \max_{\rho \in \mathcal{S}} \rho \cdot \mathbf{Q}(t),
\label{back}
\end{equation}
where $\mathbf{u}\cdot \mathbf{v} = \sum_{i = 1}^N u_i v_i$. The backpressure algorithm chooses the set of links to activate solely on the basis of current queue length and does not need to learn other parameters.

\section{Softpressure}
The backpressure algorithm ensures stability by activating the links with largest queue length. The activation is binary according to (\ref{back}). Dynamic scheduling, on the other hand, provides no such assurance of stability. In this section, we introduce an integration of schedule-driven control and backpressure that overcomes this deficiency and produces a stable schedule for a general queueing network.

\subsection{Weighted Cumulative Delay}
As mentioned earlier, a queueing network is considered to be \textit{stable} if the queues do not tend to increase without bound. According to the backpressure algorithm, serving the set of non-conflicting links with the largest queue length establishes this property. Larger queue length implies higher priority. Similarly, in the schedule-driven approach discussed earlier, we can introduce weights into this delay computation as a way to prioritize jobs from different links. The jobs with larger weights are served first. The basic idea is to bias the scheduling search more toward stabilizing local queues (both at the local nodes and at its neighbor nodes) as the level of local congestion increases. To measure the level of congestion, we rely on queue-length information associated with different links. To provide a low complexity scheme for queue management, we propose to weight each job of a given link equally. The weight associated with job $n$ on link $(i,j)$ can be expressed as
\begin{equation}
d(n) = (ast - arr(n)) \cdot w_{ij}.
\end{equation}
In the traffic signal control problem, the delay incurred by cluster $c$ is thus rewritten as
\begin{equation}
d(c) = |c| \cdot (ast - arr(c)) \cdot w(p),
\end{equation} 
where $w(p) $ is the weight assigned to the phase $p$ that cluster $c$ belongs to. The important question then becomes: how to set the weights for competing phases. 

\subsection{Weight Functions}
In this section, we provide an example of weight functions. To be suitable for combinatoric scheduling, the weights should have the following characteristics: (a) when the queues are empty or balanced, the corresponding links should be served with equal priorities (combinatorial scheduling should dominate performance in this case); and (b) the influence of shorter queues should not be reduced to zero in case jobs with larger $d_{li}$ arrive and incur larger cumulative delay. A probability function is a reasonable choice when the queue activation choice is viewed as a multinomial random variable. For instance, a node needs to pick a link to serve from $k$ links, with corresponding probabilities $p_1,\cdots, p_k$ and $\sum_{i = 1}^kp_i = 1$. If a link has larger weight, it suggests that the link should be served with a higher probability. Furthermore, the probability is a continuous function with parameters that match  the variability in processing times.

We propose a softmax function  of queue length as the weight function. The softmax formula can be derived analytically, based on a graphical model and statistical physics, if queue lengths are taken as the parameters of the graphical model \cite{Hu2017}. Assume that node $i$ has a set of queues $\{Q_{si} | s\in \mathcal{N}^{in}_{i}, (s,i)\in E\} $, where $ \mathcal{N}_i^{in}$ are the neighbors of node $i$ with a directed link $(s,i)$.  
\begin{align}
&w_{si} = \frac{ \exp(Q_{si})}{ \sum_{j\in \mathcal{N}^{in}_{i}}\exp(Q_{ji}) },
\label{multi}
\end{align}
\begin{align}
&\sum_{s \in \mathcal{N}^{in}_{i} }w_{si} = 1,
\end{align}
 For example, node $C$ in Figure~\ref{net} has two queues $Q_{AC}$ and $Q_{BC}$. The weight $w_{AC}$ is $\frac{\exp(Q_{AC})}{\exp(Q_{AC}) +  \exp(Q_{BC})}$. The softmax function fulfills the two characteristics we propose for the weight functions. However, we only takes individual node into consideration when calculating the weight function. It does not include the effect from its upstream and downstream nodes. Hence, we propose another modification of the weight functions that incorporate non-local observation.
 
\subsection{Coordinated Weight Functions}
 
Equation (\ref{multi}) only applies local queue-length information to weight the incoming jobs.  We can use non-local observation of queue-lengths of neighboring intersections to extend the prediction of queue length and improve the stability (\ref{stb}) further.  In addition to upstream neighbors $\mathcal{N}^{in}_{i}$, we denote downstream neighbors as $\mathcal{N}^{out}_i$.  Then, we define effective queue length as
\begin{align}
&\hat{Q}_{si} = Q_{si} + \sum_{h \in \mathcal{N}^{in}_{s}} Q_{hs} w_{hs} \eta^s_{hi} -  \sum_{k \in \mathcal{N}^{out}_{i}} Q_{ik} w_{ik}\eta^i_{sk} ,
\label{eff} 
\end{align}
where the $\eta^s_{hi}$ is the proportion of jobs that are routed from node $h$ to node $i$ through $s$. The second term can be viewed as the pressure that "pushes" jobs along in the $(s,i)$ direction.  The stronger the value, the greater the tendency for jobs to keep moving. The third term is analogous to a repulsive force that prevents jobs from approaching further.    

The weights for a grid network is a special case and can be derived from a probabilistic graphical model by applying the naive mean field method. Note that the weights can be computed in a decentralized way. We assume that each node knows its neighbor nodes and is able to communicate with them. First, the scheduling agent collects its local queue-length information. Once the queue-length information and the calculated weights are received from neighbor intersections, the agent then computes its weights and applies them as its job weights for generating the schedule. In the following sections, softpressure is realized as this special case: 

\begin{align}
&w_{si} = \frac{ \exp(\hat{Q}_{si})}{ \sum_{j\in \mathcal{N}^{in}_{i}}\exp(\hat{Q}_{ji}) }.
\label{hatq}
\end{align}

\subsection{Theoretical Guarantees of Stability}
In this section, we prove that by applying this weight function to jobs, an upper bound on the expected queue length is achieved. According to Little's law (queue length is equal to arrival rate multiplied by waiting time) \cite{little1961proof}, the delay is bounded as well. Furthermore, scheduling is able to provide a tighter bound than simply applying backpressure. We state the following theorem of our algorithm with the weight function (\ref{multi}).

\begin{theorem}
Consider a network has $N$ queues with arrival rates $\bm{\lambda} = (\lambda_1,\cdots, \lambda_N)$. Under the proposed softpressure algorithm, expected queue length is bounded by
\begin{equation}
 \limsup_t \mathbb{E}[\sum_{i = 1}^{N}   Q_i(t)] \leq \frac{N^2 }{2\epsilon}
 \label{bound}
 \end{equation}
 if for any queues the arrival rates satisfy $\bm{\lambda} \leq \sum_{j= 1}^K \alpha_i \bm{s}_j$ with $\sum_{j= 1}^K \alpha_j = 1 - \epsilon$ , $\epsilon > 0$ and  activation vector $\bm{s}_j \in \mathcal{S}$, $|\mathcal{S}| = K$.
\end{theorem}

\begin{proof}
Let the queue dynamics follow $Q_{i}(t+1) = Q_{i}(t) - s_i(t) \bm{1}_{Q_{i}(t) >0} + a_i(t) = Q_{i}(t)  + \Delta_i(t)$, where $s_i(t)$ and $a_i(t)$ are service and arrival rate. Note that $\mathbb{E}[s_i(t)] = \mu_i$ and $\mathbb{E}[a_i(t)] = \lambda_i$. We define the Lyapunov function $L(\bm{Q}(t)) = \sum_i Q_i^2(t)$ and use Lyapunov-Foster theory to write down the expected drift
\begin{align}
\nonumber &\mathbb{E}[L(\bm{Q}(t+1))  - L(\bm{Q}(t))| \bm{Q}(t)] \\
\nonumber&= 2 \sum_{i = 1}^N \mathbb{E}[Q_i(t+1) \cdot \Delta_i(t) | Q_i(t)] + \sum_{i = 1}^N \mathbb{E}[\Delta_i^2(t) | Q_i(t)]
\end{align}

Since the larger weight causes jobs to be serviced with higher priority until the queue is cleared, softpressure establishes the same drift criteria as backpressure: 
\begin{align}
\nonumber &\mathbb{E}[L(\bm{Q}(t+1)) | \bm{Q}(t)] \leq L(\bm{Q}(t))  - \frac{2\epsilon}{N} \sum_{i = 1}^{n} Q_i(t) + N.
\end{align}
The corresponding Lyapunov moment bound is 
\begin{equation*}
 \limsup_t \mathbb{E}[\sum_{i = 1}^{n}   Q_i(t)] \leq \frac{N^2 }{2\epsilon}.
\end{equation*}

Furthermore, as the queues become balanced or empty, the schedule-driven approach improves the service rate further through combinatorial optimization, which means that $\epsilon_{S} \geq \epsilon_{B}$,where $\epsilon_{S}$ and $\epsilon_{B}$ represent the difference between service rate and arrival rate of softpressure and backpressure respectively.
\end{proof}

\section{Empirical Investigations}
To evaluate our approach, we simulate performance on a real world network with $2$-way, multiple lane, and multi-directional traffic flow. The network model is based on the Baum-Centre neighborhood of Pittsburgh, Pennsylvania as shown in Figure~\ref{surtracmap}. The network consists mainly of 2-phased intersections. It can be seen as a two-way queueing grid network. All simulation runs were carried out according to a realistic traffic pattern from late afternoon through "PM rush" (4-6 PM). The traffic pattern ramps up volumes over the simulation interval as follows:  (0-30mins: 236 cars/hr, 30min-1hr: 354 cars/hr, 1hr-2hrs: 528 cars/hr ). This simulation model presents a complex practical application to verify the effectiveness of the proposed approach.

The simulation model was developed in VISSIM, a commercial microscopic traffic simulation software package. 
To assess the performance boost provided by our composite softpressure approach, we measure the average waiting time of all vehicles over 5 runs and take the performance of the original schedule-driven traffic control system \cite{Xie2012,xie2012schedule} as our baseline system.

\begin{table}[!htbp]
\centering
\scalebox{0.8}{
\begin{tabular}{  l| c | c  } 
  & Mean  (s) &  Std. deviation    \\  \hline
Schedule-driven & 125.77 & 111.41  \\  \hline
Softpressure & 86.53 & 64.24 \\ \hline
Backpressure & 100.23 &  76.23 
\end{tabular}
}
\caption{Avg. delay of Baum Centre Model at PM rush}
\label{resultstable2}
\end{table}

\begin{table}[!htbp]
\centering
  \scalebox{0.7}{
  \begin{tabular}{*{7}{c}}
   \toprule
      \multirow{2}{*}{}& \multicolumn{6}{c}{ Average Delay (second)}  \\
    
   \cmidrule(l){2-7}    & \multicolumn{2}{c}{Schedule-driven} &  \multicolumn{2}{c}{Softpressure}   &\multicolumn{2}{c}{ Backpressure}   \\
   & mean  & std. & mean & std. & mean &std. \\
    \midrule
%

 High demand& 212.14 & 361.41& 107.85&  77.13 & 122.58 &89.13\\
   Medium demand & 84.22& 61.90& 75.56 &55.84& 82.46 &61.40 \\
   Low demand& 71.84 &54.25& 65.10 & 49.11 & 73.89 & 56.77 \\

    \bottomrule
  \end{tabular}
  }
   \caption{Average delay under different scenarios.}
  \label{demandtable}
\end{table}

To also compare softpressure with backpressure, we developed a backpressure algorithm for traffic signal control according to \cite{varaiya2013max}. To apply the algorithm, the pressure of each intersection approach is first calculated as the corresponding queue length. The pressure of each phase is calculated as the sum of the pressures of all incoming traffic that receive right of way during this phase, and the resulting pressures are used to make a control decision. The control decision involves either a 2-second extension of the currently active phase or the activation of the next phase in the sequence. This decision process is then repeated at the end of the green extension of the currently active phase or after the minimum green time of the next phase, whichever has been decided. 
 

Table~\ref{resultstable2} shows the results of softpressure under PM rush, compared to original backpressure algorithm and the baseline schedule-driven approach. The weight function combines both local and non-local queue information. Furthermore, it incorporates the weights from neighbors.  As can be seen, the delay is reduced by $30\%$ and $15\%$, compared to the schedule-driven and backpressure approaches respectively. The use of weights can efficiently stabilize queues by controlling their length and reduce the variance of delay. The use of the neighbor information is also beneficial to get smaller queues  since it can avoid the spillover effect \cite{daganzo1998queue} by stopping vehicles further away from entry into a road segment with insufficient capacity. In addition, when upstream intersections send more traffic, the corresponding phase in the downstream intersection wills have longer green time to deal with it.

Since the evaluation focuses on highly congested scenarios, knowing the distribution of delay to vehicles helps us verify the effectiveness of the softpressure scheme. As shown in Figure~\ref{cdf}, use of softpressure with both non-local and local queue-length information shifts the cumulative distribution function (CDF) leftward and provides a $40\%$ improvement over the schedule-driven approach for $90\%$ of the vehicles. It should also be noted that while softpressure reduces average delay by $40s$, the reduction is more than $100s$ for the congested vehicles. In other words, minimizing queue length is especially effective for high congestion scenarios. In comparison to backpressure, softpressure provides a $10\%$ reduction for $90\%$ of the vehicles. 
Likewise, we can also observe a similar phenomenon if we categorize traffic demand into three different groups: high (528 cars/hr), medium (354 cars/hr), and low (236 cars/hr).  Table~\ref{demandtable} shows softpressure to yield an improvement over the schedule-driven approach of about $50\%$ for the high traffic demand case, compared to the $40\%$ improvement of backpressure.  On the other hand, performance in the medium and low traffic cases is comparable with that of the schedule-driven approach. Under medium and low traffic demand conditions, the advantage of scheduling dominates performance, and the benefit of softpressure is marginal. 

Table~\ref{demandtable2} provides another perspective on performance, using measurements of queue length.
We list the measurements of all intersections whose average queue length is greater than $2$ vehicles in the high traffic demand scenarios. The average queue length of softpressure is less than benchmark's, and thus delay is lower according to Little's law \cite{little1961proof} of queueing theory. For the heaviest loaded intersections, e.g., Baum-Aiken and Centre-Aiken, the reduction of queue length could be up to $10$ times. In addition, backpressure and softpressure have comparable average queue length. These results indicate that the combinatoric scheduling of softpressure provides further improvement over backpressure.

\begin{figure}[!htbp]
\centering
\includegraphics[scale = 0.28]{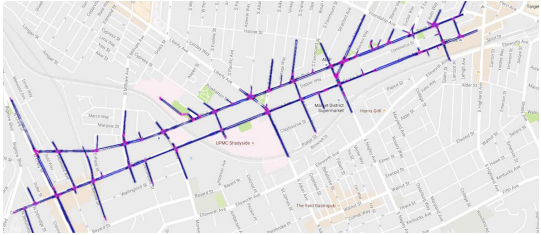}
\caption{Map of the 24 intersections in the Baum-Centre neighborhood of Pittsburgh, Pennsylvania}
\label{surtracmap}
\end{figure}

\begin{figure}[!htbp]
\centering
\includegraphics[scale = 0.32]{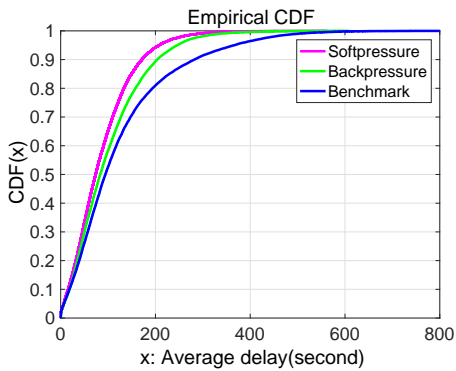}
\caption{The CDF of delay on three different optimization schemes.}
\label{cdf}
\end{figure}

\begin{table}[!htbp]
  \centering
 \resizebox{0.8\columnwidth}{!}{%
  \begin{tabular}{*{7}{c}}
   \toprule
      \multirow{2}{*}{}& \multicolumn{6}{c}{ Avg. Queue length (no.)} \\
    
   \cmidrule(l){2-7}   & \multicolumn{2}{c}{Benchmark} &  \multicolumn{2}{c}{Softpressure}   &\multicolumn{2}{c}{ Backpressure}  \\
   & mean  & std. & mean & std. & mean &std. \\
    \midrule
   Baum-Aiken & 16.01 &16.71& 7.02 & 9.32 & 7.20 &10.02 \\
   Baum-Craig & 3.34 &3.23 &2.66 & 2.49&3.55 & 3.11 \\ 
   Baum-Cypress & 2.61& 3.14 & 2.75 & 3.04 &2.82 & 3.83  \\
   Baum-Graham & 6.29 &18.43 & 2.66 & 2.76&2.31& 3.32  \\
   Baum-Millvale & 6.83 & 5.96& 6.26 &4.81 & 5.93 & 4.32 \\	
   Baum-Liberty & 16.86 &12.34&15.21 &10.95&16.56 &10.73 \\
   Baum-Melwood & 16.59& 22.36 &6.84 & 5.99 &4.15& 5.64 \\
   Baum-Negley & 7.87 & 6.66&6.45 & 5.34 &8.12 & 6.89 \\
   Baum-Roup & 5.72& 5.80 &4.19& 4.66 &6.56 & 13.43 \\
   Centre-Aiken & 41.05& 46.25 &3.24 &3.04& 3.22 & 2.97\\
   Centre-Craig & 4.93 &5.39 &4.17 & 3.84&4.83  &5.58 \\
   Centre-Cypress & 2.39 &2.37 &3.00& 2.96 &2.33 & 2.51  \\
   Centre-Millvale & 2.54 &2.73 &2.18 &2.17&2.19  & 2.24 \\
   Centre-Morewood & 4.32 &3.69&3.44 & 4.54 &3.56 & 2.85 \\
   Centre-Negley & 5.60& 4.97&5.07 &3.82 &5.60 & 4.17  \\
   Centre-Neville & 2.32 &2.42 &2.21 &2.19 &2.39 & 2.44\\
    \bottomrule
  \end{tabular}
  }
   \caption{Queue length and cluster size of the intersections under high demand traffic}
  \label{demandtable2}
\end{table}

\section{Discussion}
The motivation for studying the integration of queueing stability and scheduling is that these two communities address similar network scheduling problems in different ways and offer unique strengths. The advantages of the backpressure of queueing systems include guarantees of stability and better long-run average performance. In contrast, dynamic scheduling approaches focus on optimizing short-run objective and dealing with the variability of traffic in a periodic way. In this work, we have developed a hybridization of these two approaches to provide a better solution to a network-level scheduling problem. We have shown that  it is possible to ensure stability in a dynamic scheduling approach and keep its optimization capability.

From the experimental results, we can see that the stability of the queueing network is crucial to performance. Stability ensures that the system returns to a stable state after serving a high load and avoids system collapse \cite{stolyar2004maxweight}. In the queueing literature, establishing the stability of a system has always been the first priority in defining a control strategy. This explains why softpressure and backpressure achieved markedly stronger performance than the schedule-driven approach in the highly congested scenarios of our experimental study.  

It was also observed that softpressure and backpressure yield similar average queue lengths across all intersections in similar experimental scenarios. Yet, the delay difference between the two approaches varied up to $15\%$.  This improvement mainly results from the combinatorial optimization of traffic variability that softpressure aso includes. For example, when the queues are balanced, scheduling actively sequences jobs based on their variability, instead of treating each job equally, and sequencing jobs based on their variability leads to a shorter delay. From a theoretical perspective, we have also demonstrated that stability provides a bound on the queues (delay), and that scheduling improves this bound.

\section{Conclusions}
In this work, we described a composite approach designed to gracefully incorporate queueing stability into a network-level scheduling problem as the level of congestion increases. The approach stabilizes the queues through the use of queue-length information. This information is used to establish weights for jobs appearing on certain links. The weight function is specified as a softmax function and can be calculated in a decentralized way. The proposed approach demonstrates that theoretical long-run stability can be obtained within a dynamic scheduling approach to network control.
 The composite system was evaluated on a simulation model of a real-world traffic signal control problem. Results showed that the hybridization of queueing stability and scheduling improves average delay overall in comparison to both the baseline schedule-driven traffic control approach and the backpressure algorithm, and that solutions provide substantial gain in highly congested scenarios. Future work will focus on the design of how to decompose the queueing network into several independent sets based on the routing and traffic distribution for approaching the optimality of network scheduling. 

\section*{Acknowledgements}This research was funded in part by the University Transportation Center on Technologies for Safe and Efficient Transportation at Carnegie Mellon University and the CMU Robotics Institute.

\bibliographystyle{named}
\bibliography{softpressure}

\end{document}